\documentclass[10pt, letterpaper, conference]{IEEEtran}
\IEEEoverridecommandlockouts 
\usepackage[utf8]{inputenc}
\usepackage{subfigure}
\usepackage{amsmath}
\usepackage{mathrsfs}
\usepackage{tikz}
\DeclareMathOperator*{\argmax}{arg\,max}
\DeclareMathOperator*{\argmin}{arg\,min}
\usepackage{amssymb}
\usepackage{braket}
\usepackage{amsthm}
\usepackage{amsmath,mathtools}
\usepackage{float}
\usepackage{graphicx}
\usepackage{enumerate}
\usepackage{bbm}

\usepackage{verbatim}
\usepackage{accents}
\usepackage{authblk}

\usepackage{amsmath, amssymb, bbm, xspace}
\usepackage{epsfig}
\usepackage{longtable}
\usepackage{color}
\usepackage{mathrsfs}
\usepackage{comment}
\usepackage{ifthen}
\newboolean{showcomments}
\setboolean{showcomments}{true}
\usepackage{courier}
\usepackage{xcolor}
\usepackage{todonotes}
\usepackage[framemethod=tikz]{mdframed}
\usepackage{lineno}
\newmdenv[leftmargin=\dimexpr-0.4em, innerleftmargin=0.5em,
rightmargin=\dimexpr-0.4em, innerrightmargin=0.5em,
linewidth=2pt,linecolor=red, topline=false, bottomline=false,
innertopmargin=0pt,innerbottommargin=0pt,skipbelow=0pt,skipabove=0pt,%
]{notex}

\newenvironment{note}%
{\vskip\dimexpr\dp\strutbox-\prevdepth\relax\notex\strut\ignorespaces}%
{\xdef\notetpd{\the\prevdepth}\endnotex\vskip-\notetpd\relax}

\let\oldtodo\todo

\makeatletter%
\DeclareDocumentCommand{\todo}{ O{} +g +d<> }{%
		\setlength{\marginparwidth}{1.5cm}%
	\IfNoValueTF{#2}{\relax}{%
		\oldtodo[caption={#2},size=\scriptsize,#1]{\renewcommand{\baselinestretch}{0.8}\selectfont\sffamily#2\par}%
	}%
	\IfNoValueTF{#3}{\relax}{%
		\IfNoValueTF{#2}{
			\begin{note}%
				\begin{internallinenumbers}%
					\indent%
					#3%
				\end{internallinenumbers}%
			\end{note}%
		}{
			\vspace{-0\baselineskip}%
			\begin{note}%
				\begin{internallinenumbers}%
					\indent%
					#3%
				\end{internallinenumbers}%
			\end{note}%
		}%
	}%
}%
\makeatother
\usepackage{soul}

\newcommand{\hlc}[2][yellow]{{%
		\colorlet{foo}{#1}%
		\sethlcolor{foo}\hl{#2}}%
}

\newcommand{\removetodo}[2]{\todo[color=pink]{\textbf{delete:} ``#1'' #2}\hlc[pink]{#1}}
\newcommand{\inserttodo}[1]{\todo[color=green!40]{\textbf{insert:} #1}}

\newcommand{\hltodoy}[2]{\todo[color=yellow!40]{#2}\hl{#1} }

\newcommand{\hltodoc}[3]{\todo[color=#3!40]{#2}\hlc[#3]{#1} }

\newcommand{\hltodo}[2]{\todo[color=orange!40]{#2}\hlc[orange!40]{#1} }
\newcommand{\replacetodo}[2]{\todo[color=pink!40]{\textbf{replace with:}``#2'' }\hl{#1} }

\usepackage{marginnote}
\newcommand{\todol}[1]{{%
		\let\marginpar\marginnote
		\reversemarginpar
		\renewcommand{\baselinestretch}{0.8}%
		\todo{#1}}}

\newcommand{\inserttodol}[1]{{%
		\let\marginpar\marginnote
		\reversemarginpar
		\renewcommand{\baselinestretch}{0.8}%
		\inserttodo{#1}}}

\newcommand{\removetodol}[2]{{%
		\let\marginpar\marginnote
		\reversemarginpar
		\renewcommand{\baselinestretch}{0.8}%
		\removetodo{#1}{#2}}}

\newcommand{\hltodol}[2]{{%
		\let\marginpar\marginnote
		\reversemarginpar
		\renewcommand{\baselinestretch}{0.8}%
		\hltodo{#1}{#2}}}

\newcommand{\replacetodol}[2]{{%
		\let\marginpar\marginnote
		\reversemarginpar
		\renewcommand{\baselinestretch}{0.8}%
		\replacetodo{#1}{#2}}}

\newcommand{\hltodoyl}[2]{{%
		\let\marginpar\marginnote
		\reversemarginpar
		\renewcommand{\baselinestretch}{0.8}%
		\hltodoy{#1}{#2}}}

\newcommand{\hltodocl}[3]{{		\let\marginpar\marginnote
		\reversemarginpar
		\renewcommand{\baselinestretch}{0.8}%
		\hltodoc{#1}{#2}{#3}}}

\newtheorem{theorem}{Theorem}[section]

\newtheorem{proposition}[theorem]{Proposition}

\newtheorem{definition}{Definition}[section]

\def\bkE{{\rm I\kern-.17em E}}
\def\bk1{{\rm 1\kern-.17em l}}
\def\bkD{{\rm I\kern-.17em D}}
\def\bkR{{\rm I\kern-.17em R}}
\def\bkP{{\rm I\kern-.17em P}}

\def\bkZ{{\bf{Z}}}

\def\bkE{{\rm I\kern-.17em E}}
\def\bk1{{\rm 1\kern-.17em l}}
\def\bkD{{\rm I\kern-.17em D}}
\def\bkR{{\rm I\kern-.17em R}}
\def\bkP{{\rm I\kern-.17em P}}

\makeatletter
\newcommand{\pushright}[1]{\ifmeasuring@#1\else\omit\hfill$\displaystyle#1$\fi\ignorespaces}
\newcommand{\pushleft}[1]{\ifmeasuring@#1\else\omit$\displaystyle#1$\hfill\fi\ignorespaces}
\makeatother

\def\bkZ{{\bf{Z}}}
\def\b12{(\beta_1,\beta_2)}

\newenvironment{example}{{\noindent \bf Example}}{\hfill $\square$\hspace{-4.5pt}\vspace{6pt}}
\newcounter{example}
\renewcommand{\theexample}{\thesection.\arabic{example}}

\newcounter{remark}
\renewcommand{\theremark}{\thesection.\arabic{remark}}

\newlength{\noteWidth}
\long\def\notes#1{\ifinner
{\tiny #1}
\else
\marginpar{\parbox[t]{\noteWidth}{\raggedright\tiny #1}}
\fi\typeout{#1}}

 \def\notes#1{\typeout{read notes: #1}} 

\newcommand{\ie}{i.e.\@\xspace} 

\newcommand{\Real}{\ensuremath{\mathbb{R}}}

\def\Nbb{{\mathbb{N}}}

\def\ast{\buildrel{t}\over\rightarrow}

\def\dim{\mathop{\hbox{\rm dim}}}

\def\Span{\mbox{\rm span}}

\def\spose#1{\hbox to 0pt{#1\hss}}

\def\text #1{\hbox{\quad#1\quad}}

\def\Escr{\mathcal{E}}

\def\nthinsp{\mskip -2   mu}

\def\superstar{^{\raise 0.5pt\hbox{$\nthinsp *$}}}
\def\SUPERSTAR{^{\raise 0.5pt\hbox{$*$}}}

\def\lamstarT {\lambda^{\raise 0.5pt\hbox{$\nthinsp *$}T}}

\def\Ascr{{\cal A}}

\def\Hscr{{\cal H}}

\def\Vscr{{\cal V}}

\def\Rscr{{\cal R}}
\def\Gscr{{\cal G}}
\def\Cscr{{\cal C}}

\let\forallnew\forall
\renewcommand{\forall}{\forallnew\ }
\let\forall\forallnew

		\def\bkE{{\rm I\kern-.17em E}}
		\def\bk1{{\rm 1\kern-.17em l}}
		\def\bkD{{\rm I\kern-.17em D}}
		\def\bkR{{\rm I\kern-.17em R}}
		\def\bkP{{\rm I\kern-.17em P}}
		\def\bkY{{\bf \kern-.17em Y}}
		\def\bkZ{{\bf \kern-.17em Z}}
		\def\bkC{{\bf  \kern-.17em C}}

{\begin{list}{}%
         {\setlength{\leftmargin}{#1}}%
         \item[]%
}
{\end{list}}

\def\Rsf{\mathsf{R}}
		\def\bsp{\begin{split}}
		\def\beq{\begin{eqnarray}}
		\def\bal{\begin{align*}}
		\def\bc{\begin{center}}
		\def\be{\begin{enumerate}}
		\def\bi{\begin{itemize}}
		\def\bs{\begin{small}}
		\def\bS{\begin{slide}}
		\def\ec{\end{center}}
		\def\ee{\end{enumerate}}
		\def\ei{\end{itemize}}
		\def\es{\end{small}}
		\def\eS{\end{slide}}
		\def\eeq{\end{eqnarray}}
		\def\eal{\end{align*}}
		\def\esp{\end{split}}
		\def\qed{ \vrule height7.5pt width7.5pt depth0pt}  

	\def\cp2problem#1#2#3#4{\fbox
		 {\begin{tabular*}{0.9\textwidth}
			{@{}l@{\extracolsep{\fill}}l@{\extracolsep{6pt}}l@{\extracolsep{\fill}}c@{}}
				#1 & & $#4 $
			\end{tabular*}}}

		\def\bkE{{\rm I\kern-.17em E}}
		\def\bk1{{\rm 1\kern-.17em l}}
		\def\bkD{{\rm I\kern-.17em D}}
		\def\bkR{{\rm I\kern-.17em R}}
		\def\bkP{{\rm I\kern-.17em P}}

		\def\bkZ{{\bf{Z}}}

\newcommand {\beeq}[1]{\begin{equation}\label{#1}}
\newcommand {\eeeq}{\end{equation}}
\newcommand {\bea}{\begin{eqnarray}}
\newcommand {\eea}{\end{eqnarray}}

\def\texitem#1{\par\smallskip\noindent\hangindent 25pt
               \hbox to 25pt {\hss #1 ~}\ignorespaces}

\def\bsp{\begin{split}}
		\def\beq{\begin{eqnarray}}
		\def\bal{\begin{align*}}
		\def\bc{\begin{center}}
		\def\be{\begin{enumerate}}
		\def\bi{\begin{itemize}}
		\def\bs{\begin{small}}
		\def\bS{\begin{slide}}
		\def\ec{\end{center}}
		\def\ee{\end{enumerate}}
		\def\ei{\end{itemize}}
		\def\es{\end{small}}
		\def\eS{\end{slide}}
		\def\eeq{\end{eqnarray}}
		\def\eal{\end{align*}}
		\def\esp{\end{split}}
		\def\qed{ \vrule height7.5pt width7.5pt depth0pt}  

\usepackage{amsmath, amssymb, xspace}
\usepackage{epsfig}
\usepackage{longtable}
\usepackage{color}
\usepackage{mathrsfs}
\usepackage{subfig}
\def\Cscr{{\cal C}}

\reversemarginpar
\title{\bf Characterizing Flow Complexity in Transportation Networks using Graph Homology}
\author{Shashank A. Deshpande \and Hamsa Balakrishnan\thanks{Shashank A Deshpande is a graduate student in the Department of Aeronautics and Astronautics at the Massachusetts Institute of Technology (MIT). Hamsa Balakrishnan is the William E. Leonhard (1940) Professor of Aeronautics and Astronautics at MIT. This work was sponsored in part by NASA grant \#80NSSC23M0220 and the NASA University Leadership Initiative (grant \#80NSSC21M0071), but this article solely reflects the opinions and conclusions of its authors and not NASA.}} 
\begin{document}
	\maketitle
	\begin{abstract}
Series-parallel network topologies generally exhibit simplified dynamical behavior and avoid high combinatorial complexity. A comprehensive analysis of how flow complexity emerges with a graph's deviation from series-parallel topology is therefore of fundamental interest.  We introduce the notion of a robust $k$-path on a directed acycylic graph, with increasing values of the length $k$ reflecting increasing deviations. We propose a graph homology with robust $k$-paths as the bases of its chain spaces. In this framework, the topological simplicity of series-parallel graphs translates into a triviality of higher-order chain spaces. We discuss a correspondence between the space of order-three chains and sites within the network that are susceptible to the Braess paradox, a well-known phenomenon in transportation networks. In this manner, we illustrate the utility of the proposed graph homology in sytematically studying the complexity of flow networks.
	\end{abstract}
 
	\section{Introduction}
Directed graphs are widely used to model flows in many real-world transportation and logistic networks. A directed acyclic graph (DAG) is said to be a series-parallel graph if it can be constructed via sequential series and parallel combination of edges or smaller series-parallel graphs. The possession of a series-parallel topology is a global property of a DAG; it is not thoroughly characterized by localized subgraphs within the graph. Series-parallel topologies are a form of topological simplicity: for example, a many combinatorial problems can be solved in linear-time on series-parallel graphs~\cite{takamizawa1982lineartime}. Similarly, series-parallel topologies are known to simplify the analysis of electrical networks~\cite{cederbaum1984graphnetanalysis}. Therefore, the deviation of a graph from a series-parallel topology can be considered an increase in its flow complexity. 

Another example of such flow complexity is the occurance of the Braess Paradox, a well-known phenomenon in flow networks (e.g., transportation \cite{RoughgardenTardos2002}, power grids~\cite{Schafer2022powgrid}, ecological networks \cite{sahasrabudhe2011ecobraess}) where the addition of a link results in the slowdown of flows. Prior studies have shown that the Braess Paradox occurs when a network deviates from a series-parallel topology~\cite{milchtaich2006topology, chen2016netchar}. 
 Localizing `sites' within a network that are susceptible to the Braess paradox is therefore of broad relevance in network analysis. This paper seeks to systematically characterize the flow complexity that arises as a consequence of a network deviating from a series-parallel topology.

 We introduce a notion of robust path of length $k$ (or a robust $k$-path) on a DAG, where increasing length is a reflection of increased flow complexity. For instance, we find that the presence of a robust $3$-path is a necessary and sufficient condition for a network to deviate from a series-parallel topology, and that each robust $3$-path is associated with a site susceptible to the Braess Paradox. Further, a robust $k$-path identifies the presence of upto ${k\choose 4}$ distinct susceptible sites within the network. This motivates us to develop a systematic approach for the characterization of flow complexity in DAGs using robust paths as basic objects. For this purpose, we utilize the algebraic structure of graph homology. In particular, we associate the linear spans of robust $k$-paths with $k$-chains in a graph homological framework and prove that the association sets up a consistent chain complex. We demonstrate that the induced chain complex provides a representation of the underlying DAG where higher-order chains identify sites of high flow complexity within the graph. We  illustrate the utility of this framework by showing that series-parallel topology of a DAG translates to triviality of $3$-chains in the chain complex. This algebraic restatement of a known combinatorial result is validation of how the proposed homology can be used to systematically investigate flow complexity.

 Homological approaches have been previously used to study global features in complex real-world networks. For example, simplicial homology has been deployed as a generalized clustering mechanism that identifies interconnections within and among clustered communities on undirected graphs~\cite{petri2013strata, aktas2019perhomonet}. Path homology~\cite{grigoryan2020homologies} on directed graphs has been shown to identify topological characteristics that classify complex networks\cite{chowdhury2022pathhomo, chowdhury2019deepnet}, although the intuition behind this classification remains largely intractable. A prior interpretation that path homology measures the consistency and robustness of directional flow in a graph is in line with the graph homology we develop in this article \cite{chowdhury2022pathhomo, chowdhury2019deepnet}.
  We believe that our approach can be used for the systematic localization of flow complexity in networks, and to understand its implications. 

    The organization of this article proceeds as follows. In the brief subsection that follows, we introduce the concepts of series-parallel graphs and robust paths, and the role of the latter in reflecting the deviation of a flow from the series-parallel nature. In Section \ref{sec:graphhomology}, we develop a consistent algebraic structure for the formal study of these concepts. Subsequently in Section \ref{sec:spalgebra}, we formalise the notion of a series-parallel topology and use the developed structure to produce an algebraic characterization of the same, as well as to characterize deviations from this topology. 

    \begin{figure*}[ht!]
    \centering
    \subfigure[]{\includegraphics[width=6cm]{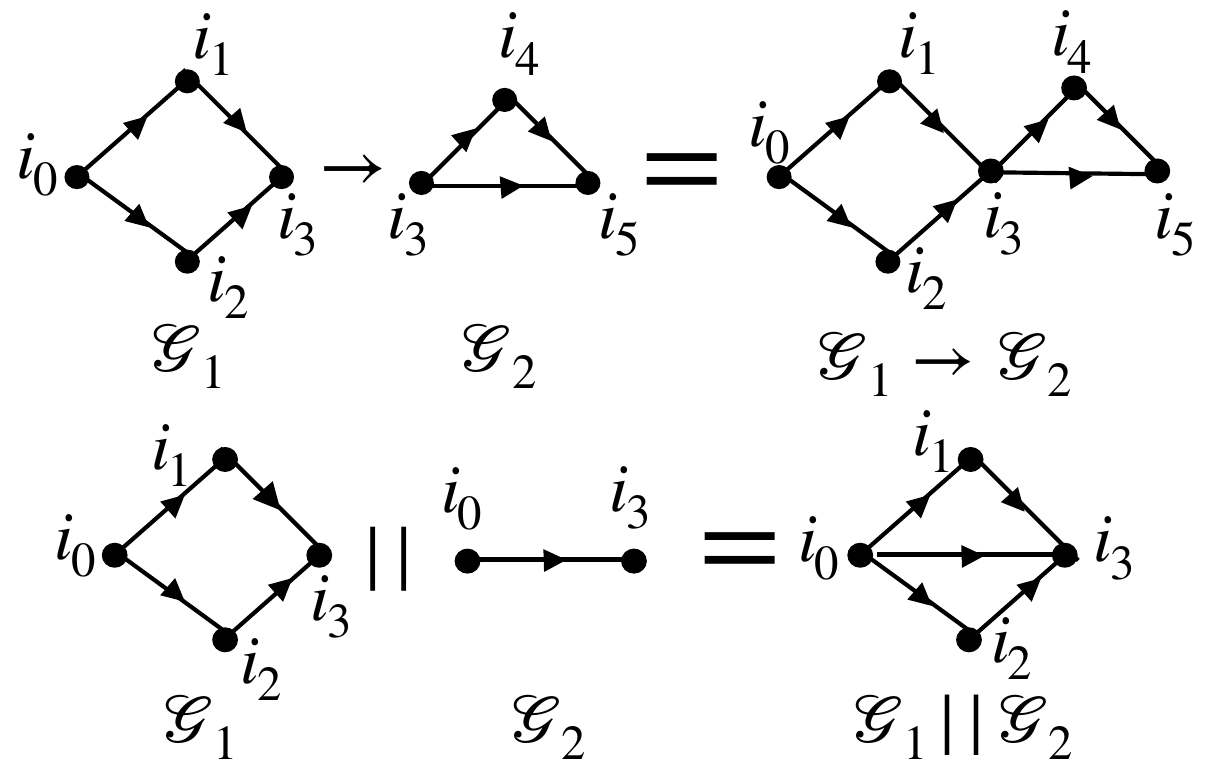}} 
    \subfigure[]{\includegraphics[width=4cm]{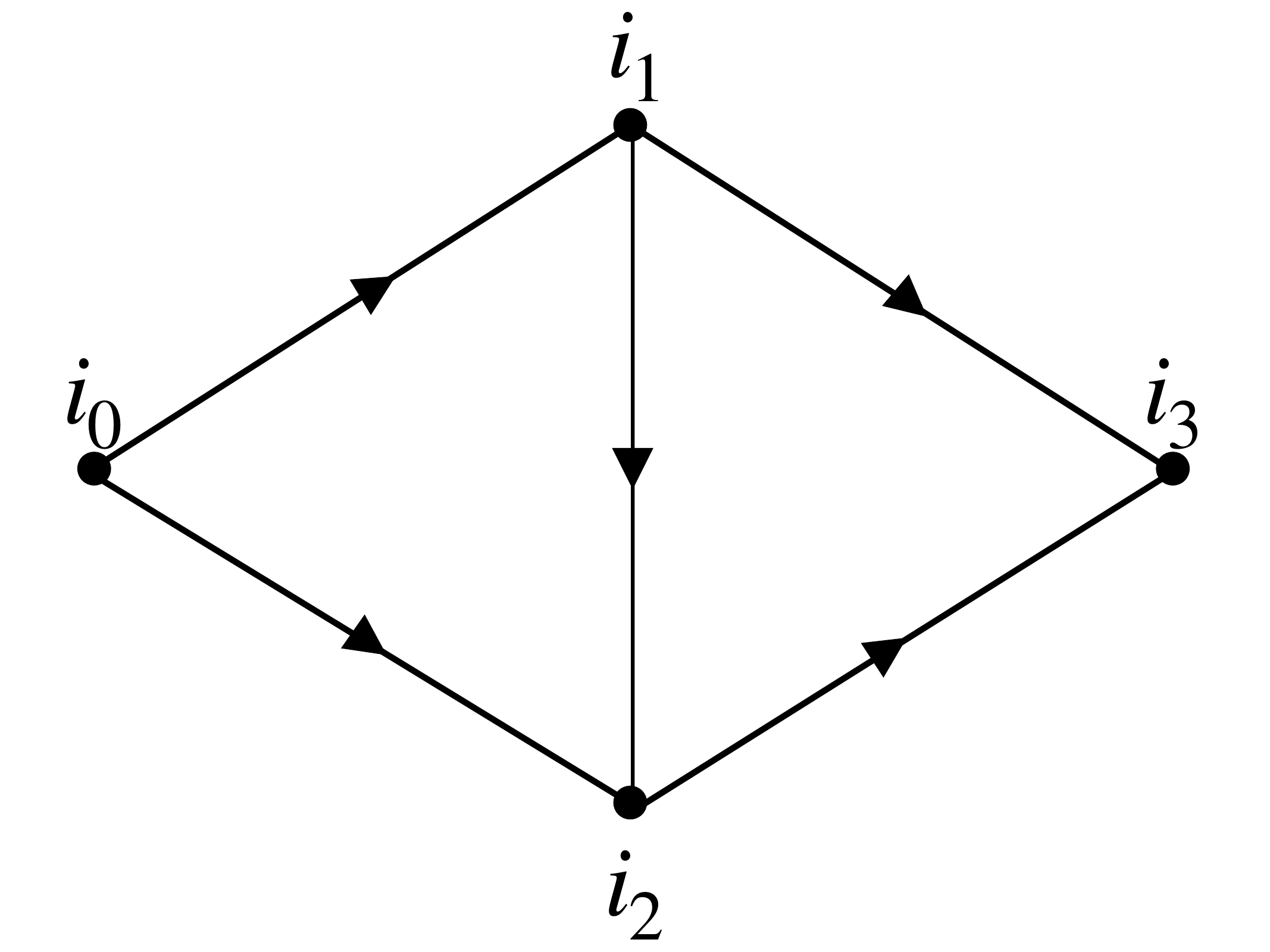}}
    \subfigure[]{\includegraphics[width=7cm]{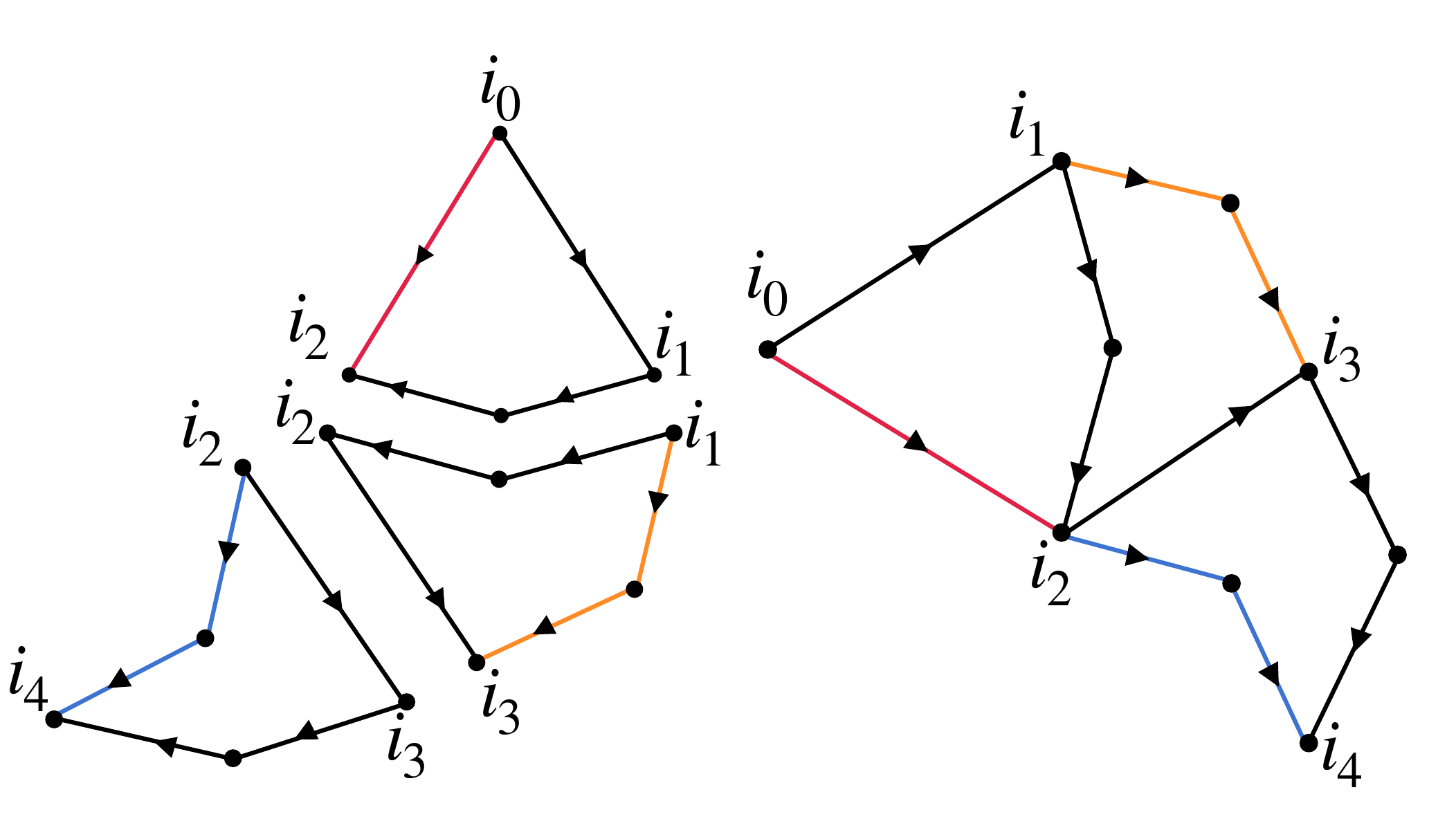}}
    \caption{(a) Series and Parallel Combination (b) The Braess Embedding (c) Robust 2-paths combine into a robust 4-path }
    \label{fig:collection1} 
\end{figure*}
        \subsection{Series-parallel graphs and the Braess embedding}
        We are interested in a class of DAGs called two-terminal graphs where directional flows emanate from an origin vertex (source) and are absorbed by a destination vertex (sink). 
        Series-parallel graphs are two-terminal graphs obtained by serially or parallely combining edges and/or smaller series-parallel graphs. See Figure \ref{fig:collection1}(a) for a depiction of series and parallel combination operations.  The departure of a two-terminal graph from the series-parallel topology is known to follow from the appearance of the structure called the Braess embedding or a Braess site, shown in Figure \ref{fig:collection1}(b) as a graphical embedding within the network~\cite{milchtaich2006topology, chen2016netchar}. The tuple of the vertices involved in the embedding (e.g., $(i_0, i_1, i_2, i_3)$ in Figure \ref{fig:collection1}(b)) localizes the site within a network.
 
 We wish to investigate the deviation of a graph from a series-parallel topology in a comprehensive manner. To this end, we introduce the notion of a robust $k$-path in a DAG. The basic object in our discussion is the robust 2-path, which we also call a \emph{triangle}.  We call $e_{i_0i_1i_2}$ a robust 2-path in $\Gscr$ if  $i_0, i_1, i_2$ are three vertices in $\Gscr$, and there exists a pair of non-intersecting routes from $i_0$ to $i_2$, exactly one of which passes through through $i_1$. If robust 2-paths occur as adjacent structures within the graph, they give rise to longer robust paths. Therefore, $e_{i_0...i_p}$ is a robust $p$-path if $e_{i_k i_{k+1} i_{k+2}}$ is a robust $2$-path for each $k\in [p-2]$, . For illustration, see Figure \ref{fig:collection1}(c) where three adjacent robust 2-paths $(e_{i_0i_1i_2}$, $e_{i_1i_2i_3}$, $e_{i_2i_3i_4})$ are shown to merge and give rise to a robust 4-path $(e_{i_0i_1i_2i_3i_4})$.
 
We will show that the presence of a robust 3-path ensures the presence of a Braess-susceptible site or a Braess embedding within the network.
More generally, long robust paths within a network contribute to the rising flow complexity as the network topology deviates from a series-parallel one. 

We associate that linear spaces spanned by robust paths to chains within an algebraic structure of graph homology. We expect that the algebraic structure can be leveraged to systematically study flow complexity and related features in two-terminal graphs. 
To this end, we show that series-parallel graphs associate with a chain complex truncated at order two, that is, robust paths of length three and above are absent in the associated chain complex. Therefore, we map the topological simplicity of series-parallel graphs onto an algebraic specification in the graph homology of robust paths. We believe that the constructed graph homology carries further potential for the characterization and analysis of complex features in flow networks.
	 \subsection{Notation}
	 \begin{enumerate}[(i)]
	 	\item For a DAG $\Gscr=(\Vscr, \Escr)$, we denote a directed edge from $i\in\Vscr$ to $j\in\Vscr$ by $e_{ij}$. 
        \item  $e_{i_0...i_p}$ is used to denote the tuple $(i_0,...,i_p)\in \Vscr^{p+1}$.
        \item  $i\in\Gscr$ and $e_{ij}\in\Gscr$ respectively mean $i\in\Vscr$ or $e_{ij}\in \Escr$. 
	 	\item $[N]$ denotes the set $\{1, 2, ..., N\}$ for each $N\in\Nbb$.
	 	\item Union and intersection on graphs are as usual, for $\Gscr_i=(\Vscr_i, \Escr_i)$, $\cup_i\Gscr_i=(\cup_i\Vscr_i, \cup_i\Escr_i)$ and $\cap_i\Gscr_i=(\cap_i\Vscr_i, \cup_i\Escr_i)$.
        \item For DAGs $\Gscr_1$ and  $\Gscr_2$, we say $\Gscr_1\cong\Gscr_2$ if $\Gscr_1=\Gscr_2$ up to relabelling of their vertices and edges.
        \item $K_{ij}$ denotes the edge graph  $K_{ij}:=(\{i, j\}, \{e_{ij}\}, i, j)$.
        \item $\mathbb{K}$ denotes a field, the reader may specialise to $\mathbb{K}=\Real$.
	 \end{enumerate}
	
	\section{Graph Homology of Robust Paths}\label{sec:graphhomology}
    In Subsection \ref{sec:graphdef} below, we define routes, two-terminal graphs, and colored route simplices of two-terminal graphs. In Subsection \ref{sec:graphhomodev} that follows, we formalise the notion of a robust path and embed the linear spaces spanned by the robust paths into the homological algebra.
	\subsection{Two-terminal DAGs and colored route simplices}\label{sec:graphdef}
	A two-terminal DAG is induced by a union of linear graphs or routes, which we define as follows. We also introduce formal notation for segments of a route, which are shorter routes with different origin-destination pairs. 
	\begin{definition} 
	(i)	A route $\Rscr$ is a tuple $\Rscr=(\Vscr, \Escr, o, d, r)$ where $\Vscr$ is a finite set of nodes or vertices, $r:\Vscr\to\Nbb$ is a strict order on $\Vscr$, $o=\argmin_{i\in\Vscr} r(i)$, $d=\argmax_{j\in\Vscr} r(j)$, and, $\Escr\subset \Vscr\times\Vscr$ contains all edges $e_{ij}:=(i, j)$ if and only if $i$ and $j$ are consecutive in the order $r$, that is, $r(i)<r(j)$ and  $\nexists\ k\in\Vscr: r(i)<r(k)<r(j)$.\\
		(ii) Let $\Rscr=(\Vscr, \Escr, o, d, r)$ be a route and $i, j\in\Vscr$ be two of its vertices. Define and denote another route from $i$ to $j$ as follows: $\Rscr^{i\to j}:=(\Vscr^{i\to j}, \Escr^{i\to j}, i, j, r)$ where $\Vscr^{i\to j} =\{k\in\Vscr: r(i) \leq r(k)\leq r(j)\}$, $\Escr^{i\to j}=\{e_{ ab}| a, b\in \Vscr^{i\to j}\}\cap \Escr$. We regard $\Rscr^{i\to i}=(\{i\}, \phi, i, i, r)$ as the vertex $i$.
	\end{definition}
     \noindent\emph{Note:} Let $\Rscr_1$ and $\Rscr_2$ be two arbitrary routes. If the intersection graph $\Rscr_1\cap\Rscr_2$ is non-empty, then we take note of the fact that it is expressible in the following form:
     \begin{equation}
     	\Rscr_1\cap\Rscr_2=\bigcup_{n=1}^{n_0} \Rscr_1 ^{p_n\to q_n}=\bigcup_{n=1}^{n_0} \Rscr_2 ^{p_n\to q_n} \label{eq:intersectiongraph}
     \end{equation}
     where $n_0\in\Nbb$, $p_n, q_n\in \Vscr_1\cap\Vscr_2$ for each $n$.
     
     A union of routes that share the same origin-destination pair induces a two-terminal DAG as follows. Acyclicity of the induced DAG is ensured by requiring the routes to respect each other's order.
	\begin{definition}
		i) Let $\{\Rscr_\alpha\}_{\alpha\in A}=\{(\Vscr_\alpha, \Escr_\alpha, o_\alpha, d_\alpha, r_\alpha)\}_{\alpha\in A}$ be a finite collection of routes with the same origin $o$ and destination $d$ (\ie $o_\alpha\equiv o, d_\alpha\equiv d$) that obey the partial order induced by $\{r_\alpha\}_{\alpha\in A}$:
		\begin{multline}\label{eq:partialorderconsistent}
			\forall \delta, \beta\in A,  \{i, j\}\in \Vscr_\delta\cap \Vscr_\beta\implies\\
			r_\delta(i)< r_\delta(j)\longleftrightarrow r_\beta(i)<r_\beta(j)
		\end{multline}
		Then, the tuple $\Gscr:=(\Vscr:=\cup_{\alpha\in A} \Vscr_\alpha, \Escr:=\cup_{\alpha\in A} \Escr_\alpha, o, d, (r_{\alpha})_{\alpha\in A})$ is called a two-terminal graph from origin $o$ to destination $d$ induced by the collection of routes $(\Rscr_\alpha)_{\alpha\in A}$. We then write
		$\Gscr=\bigcup_{\alpha\in A} \Rscr_\alpha$ and say that $i<j$ if $i, j\in\Vscr$ and $\exists \alpha\in A$ such that $r_\alpha(i)<r_\alpha(j)$.
	\end{definition}
	\noindent \emph{Note:} We call the collection $\{\Rscr_i\}_{i\in [N]}$ a complete enumeration of routes in $\Gscr=\bigcup_{i\in[N]}\Rscr_i$ if it contains all $o$ to $d$ routes in $\Gscr$. All collections in this article are assumed to be complete enumerations. We also drop the underlying partial order $(r_i)_{i\in [N]}$ in our notation and use $\Gscr=(\Vscr, \Escr, o, d)$ to represent the two-terminal graph. 
	
	A route induces a two-terminal DAG we call a route-simplex as follows. The route-simplex shares the vertex set of the underlying route $\Rscr$ and contains an edge $e_{ij}$ if $j$ is reachable from $i$. A union of route simplices induced by the constituent routes of a two-terminal DAG is declared as the route simplex of the DAG. We attach a multi-coloring to each edge $e_{ij}$ in a route-simplex to record the set of routes that reach $j$ from $i$; this produces what we call a colored route simplex of a DAG. These notions are formalised by the definition below.
	\begin{definition}
		(i) The route-simplex of $\Rscr_i=(\Vscr_i, \Escr_i, o,d, r_i)$ denoted by $Sim(\Rscr_i)$ is the two terminal graph $(\Vscr_i, \Rsf(\Escr_i), o, d)$ where $\Rsf(\Escr)=\{e_{ij}: i, j\in \Vscr, r(i)<r(j)\}$.\\
		(ii) The route-simplex of $\Gscr=\bigcup_{i\in[N]} \Rscr_i$ is defined to be the union $Sim(\Gscr):=\bigcup_{i\in [N]} Sim(\Rscr_i)$.\\
        (iii) The colored route simplex of $\Gscr$ is the tuple $\Rsf(\Gscr):=(\Vscr, \Rsf(\Escr), o, d, \Cscr)$  where the (multi)coloring  $\Cscr: \Rsf(\Escr)\to 2^{[N]}$ obeys $\Cscr(e_{pq})=\{i\in [N]| e_{pq}\in Sim(\Rscr_i)\}$.
	\end{definition}
	Consider four routes $\{\Rscr_\alpha\}_{\alpha\in [4]}=\{(\Vscr_\alpha, \Escr_\alpha, o, d, r_\alpha)\}$ with $\Vscr_1=\{o, 1, 2, d\}, \Escr_1=\{e_{o1}, e_{12}, e_{2d}\},
		\Vscr_2=\{o, 1, 2, 3, d\}, \Escr_2=\{e_{o1}, e_{12}, e_{24}, e_{4d}\},
		\Vscr_3=\{o, 3, 4, d\}, \Escr_3=\{e_{o3}, e_{34}, e_{4d}\},
		\Vscr_4=\{o, 3, 5, d\}, \Escr_4=\{e_{o3}, e_{35}, e_{5d}\}$
	which constitute the two-terminal graph $\Gscr=\cup_\alpha \Rscr_\alpha$ shown in Figure \ref{fig:dim3unnecessary} below.
	\begin{figure}[h]
		\centering
		\includegraphics[width=7cm]{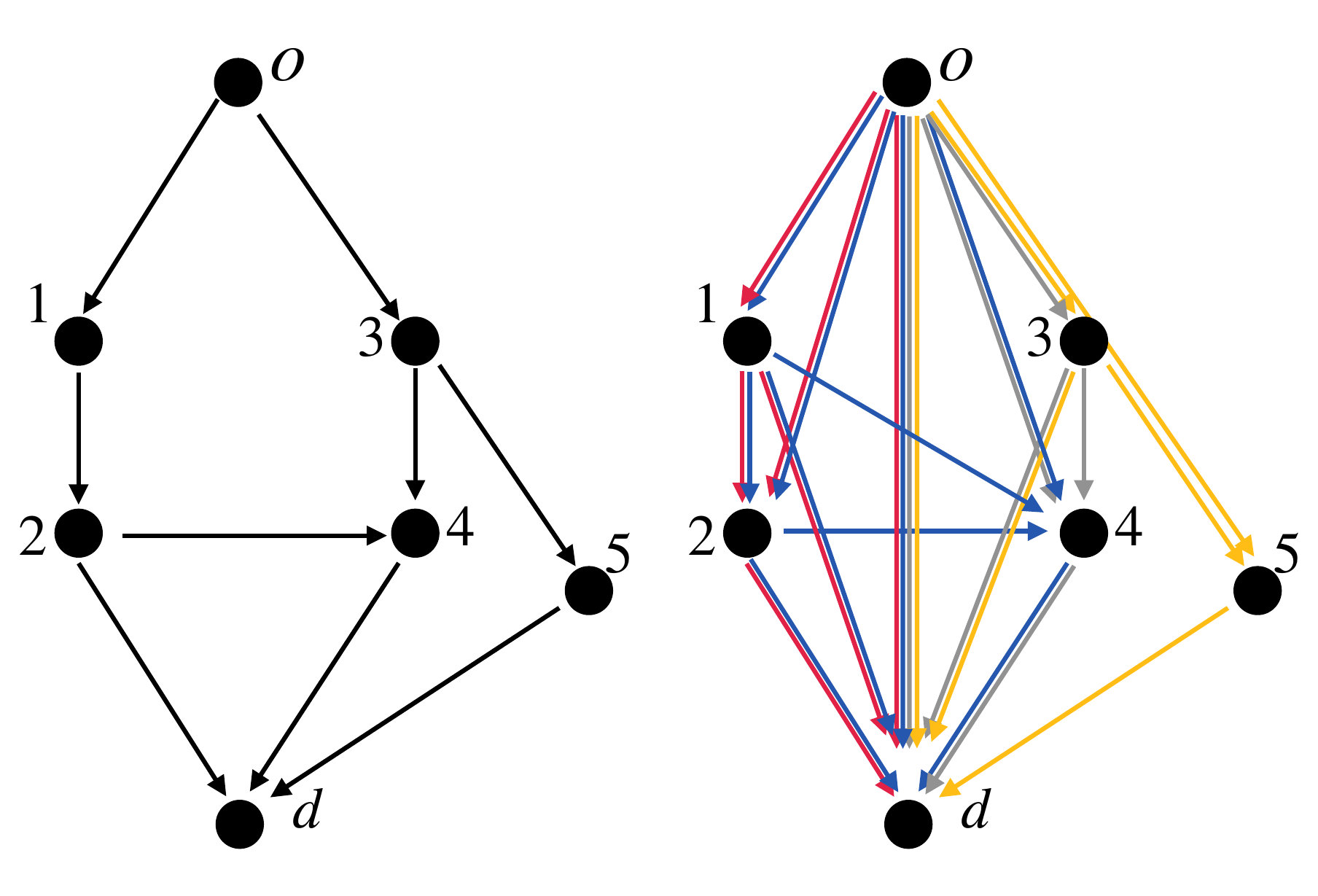}
	\caption{$\Gscr=\bigcup_\alpha \Rscr_\alpha$; $\Rsf(\Gscr): \dim\Omega_3(\Rsf(\Gscr))\neq0$}
		\label{fig:dim3unnecessary}
	\end{figure}
	 $\Rsf(\Gscr)$  is depicted alongside where Red, Blue, Gray and Yellow respectively depict the colors $1, 2, 3$ and $4$.
	\subsection{Development of the robust path homology}\label{sec:graphhomodev}
    We now develop our graph homology of robust paths.
	Let $\Gscr=(\Vscr, \Escr, o, d)=\bigcup_{i=1}^N \Rscr_i=\bigcup_{i=1}^N (\Vscr_i, \Escr_i, o, d, r_i)$ be a two-terminal DAG and $\Rsf(\Gscr)=(\Vscr, \Rsf(\Escr), \Cscr)$ be the colored route simplex of $\Gscr$. The space of of vertex tuples $\Vscr^{p+1}$ is refined to record graph topology in the refined subset of allowed paths.
	\begin{definition}
		(i)  $e_{i_0...i_p}$ is an elementary allowed $p$-path in $\Rsf(\Gscr)$ if $e_{i_{m-1} i_{m}}\in \Rsf(\Escr)$ for all $m\in [p]$.\\
		(ii) We define the $\mathbb{K}$-linear span of all elementary allowed $p$-paths as the space of allowed $p$-paths: 
  $$\Ascr_p(\Rsf(\Gscr)):=\mathbb K\mbox{-}\Span\{e_{i_0...i_p}: e_{i_j i_{j+1}}\in\Rsf(\Escr)\ \forall\ j\in [p-1]\}.$$
	\end{definition}
 \noindent Next, we define a linear operator on the allowed path spaces.
	\begin{definition}
		 The linear boundary operator $\partial_p: \Ascr_p(\Rsf(\Gscr))\to \Ascr_{p-1}(\Rsf(\Gscr))$ is a linear operator defined via its action on elementary paths: $\partial_p e_{i_0...i_p}=\sum_{k}(-1)^ke_{i_0...\widehat{i_k}...i_p}$
		and extended over $\Ascr_p(\Rsf(\Gscr))$ by linearity. Note that $\partial_p\equiv 0$.
		\end{definition}
	 \noindent Elementary allowed paths are further refined to exclude non-robust paths. The robust $k$-paths then become the basis set for the space of $k$-chains.	
		\begin{definition}\label{def:robustpath}
			(i) An allowed $e_{i_0i_1i_2}$ is a robust 2-path or a triangle in $\Gscr$ if there is a route from $i_0$ to $i_2$ that  evades atleast one route from $i_0$ to $i_2$ through $i_1$, \ie,  $\exists (\alpha, \beta) \in \Cscr(e_{i_0i_2})\times \Cscr(e_{i_0i_1})\cap\Cscr(e_{i_1 i_2})$  such that 
			$\Vscr_\alpha^{i_0\to i_2}\cap \Vscr_\beta=\{i_0, i_2\}$.
			 We then call the tuple of routes $(\Rscr_\alpha, \Rscr_\beta)$, a triangulating pair of the robust $2$-path $e_{i_0 i_1 i_2}$. We denote the set of all triangles (robust 2-paths) by $\Delta_2(\Rsf(\Gscr))$.\\
    (ii) An allowed $e_{i_0 ... i_p}$ is a robust p-path in $\Gscr$ if it is allowed and the 2-path  $e_{i_{k-1}i_{k}i_{k+1}}\in \Delta_2(\Rsf(\Gscr))$ for each $k\in [p-1]$. Denote the set of all robust p-paths by $\Delta_p(\Rsf(\Gscr))$.
		\end{definition}
 
		\begin{definition}\label{def:invpath}
		(i)  The sets of 0-chains and 1-chains are respectively defined as $\Omega_0(\Rsf(\Gscr)):=\mathbb K$-$\Span\{\Vscr\}=\Ascr_0(\Rsf(\Gscr))$ and $\Omega_1(\Rsf(\Gscr)):=\mathbb K$-$\Span \{\Rsf(\Escr)\}=\Ascr_1(\Rsf(\Gscr))$.\\
		(ii) The set of $p$-chains is defined as the $\mathbb K$-linear span of robust p-paths:  $\Omega_p(\Rsf(\Gscr))=\mathbb K$-$\Span\{\Delta_p(\Rsf(\Gscr))\}\subseteq\Ascr_p(\Rsf(\Gscr))$.
	\end{definition} 
 \noindent\emph{Note:} $e_{i_0...i_p}\in \Delta_{p}(\Rsf(\Gscr))\longleftrightarrow e_{i_0...i_p}\in \Omega_{p}(\Rsf(\Gscr))$.\\
 The following proposition sets up the desired homology of the introduced $k$-chain spaces $\{\Omega_k(\Rsf(\Gscr))\}_{k\in\Nbb_0}$.

\begin{proposition}\label{prop:homologysetup}
		For all $p\geq 1$, we have\\
		(i) $\partial_{p-1}\circ\partial_{p}=0$.\\
		 (ii) $\partial \Omega_p(\Rsf(\Gscr))\subseteq \Omega_{p-1}(\Rsf(\Gscr))$.\\
		 Consequently, we obtain the following chain complex
		$$\mathbb K \{0\} \xleftarrow{\partial_0} \Omega_0(\Rsf(\Gscr))\xleftarrow{\partial_1}...\xleftarrow{\partial_n} \Omega_n(\Rsf(\Gscr))\xleftarrow{\partial_{n+1}}\cdots$$
	\end{proposition}
    \begin{figure*}[t]
		\centering
		\includegraphics[width=18cm]{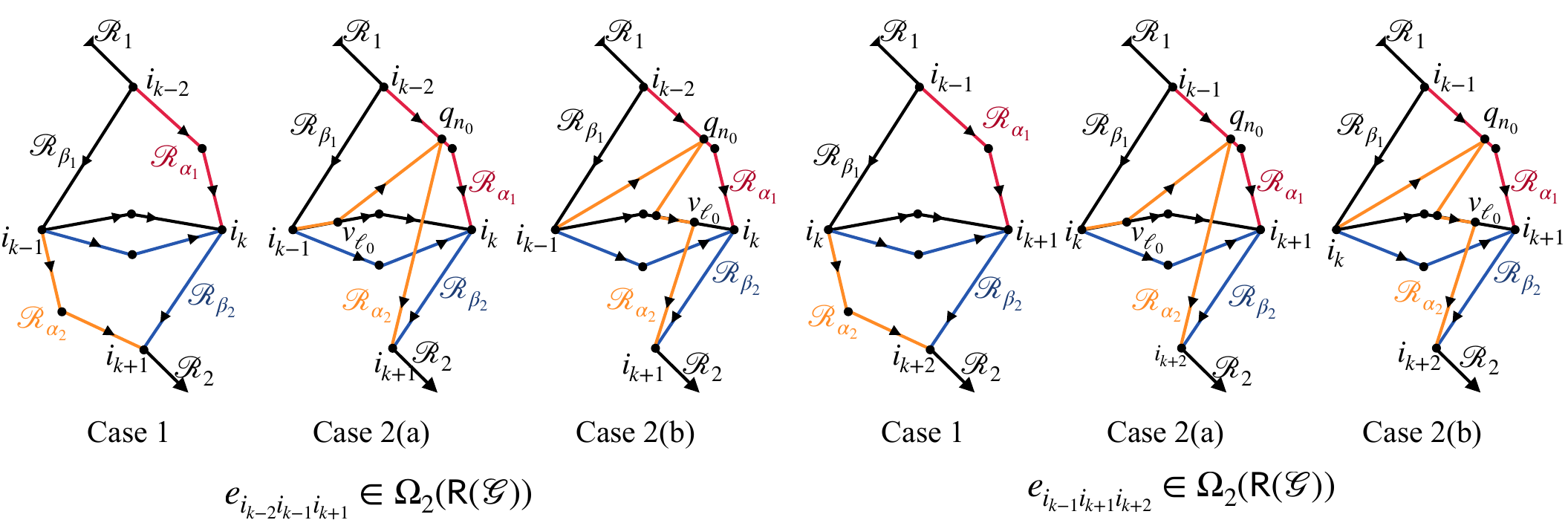}
		\caption{Representative Topologies for each case in the proof of the Proposition \ref{prop:homologysetup}}
		\label{fig:mainproof}
	\end{figure*}
	\begin{proof}
	(i) For an arbitrary $e_{i_0...i_p}$, 
		\begin{align*}
			\partial_{p-1}\circ\partial_p e_{i_0\cdots i_p}
			= &\sum_{r=0}^{q-1}\sum_{q} (-1)^{q+r} e_{i_0\cdots \widehat i_r\cdots \widehat i_q\cdots i_p}\\
			+ &\sum_{r=q+1}^{p}\sum_{q} (-1)^{q+r-1} e_{i_0\cdots \widehat i_q\cdots \widehat i_r\cdots i_p}=0.
		\end{align*}
       Therefore Proposition \ref{prop:homologysetup}.(i) follows by linearity of $\partial_p$.\\
		(ii) -- For $p=0, 1$, $\partial\Omega_0(\Rsf(\Gscr))=\mathbb{K}\{0\}$ and $\partial\Omega_1(\Rsf(\Gscr))\subset \mathbb{K}\{\Vscr\}=\Omega_0(\Rsf(\Gscr))$ follow by definition as $\partial_0=0$, $\Omega_0(\Rsf(\Gscr))=\Ascr_0(\Rsf(\Gscr))$ and $\Omega_1(\Rsf(\Gscr))=\Ascr_1(\Rsf(\Gscr))$. \\
       For $p\geq 2$, we show that $e_{i_0...i_p}\in \Omega_p(\Rsf(\Gscr))\implies e_{i_0...\widehat{i_k}...i_p}\in \Omega_{p-1}(\Rsf(\Gscr))$ for all $k\in \{0,\cdots,p\}$ which implies $\partial_p\Omega_{p}(\Rsf(\Gscr))\subset \Omega_{p-1}(\Rsf(\Gscr))$ by linearity of $\partial_p$.\\
		-- For $p=2$, notice that $e_{i_0i_1i_2}\in\Omega_2(\Rsf(\Gscr))\subset \Ascr_2(\Rsf(\Gscr))\implies e_{i_0i_1}, e_{i_1i_2}\in \Rsf(\Escr)$ which implies the existence of routes $$\Rscr_a: e_{i_0 i_1}\in Sim(\Rscr_a); \Rscr_b: e_{i_1 i_2}\in Sim(\Rscr_b).$$ It follows that $e_{i_0 i_2}\in Sim(\Rscr_c)$ where $\Rscr_c=\Rscr_a^{o\to i_1}\cup\Rscr_b^{i_1\to d}$. Thereby $\partial e_{i_0 i_1 i_2}\in\Ascr_1(\Rsf(\Gscr))\equiv \Omega_1(\Rsf(\Gscr))$.\\
		 -- Now let $p\geq 3$. To show that $e_{i_0...\widehat i_k...i_p}\in\Omega_{p-1}(\Rsf(\Gscr))$ for each $k$, following Definition \ref{def:robustpath}, \ref{def:invpath}, it suffices to show \\
   (A) $e_{i_{k-2}i_{k-1}i_{k+1}}\in\Omega_{2}(\Rsf(\Gscr))$ for all $p-1 \geq k\geq 2$ and \\
   (B) $e_{i_{k-1} i_{k+1} i_{k+2}}\in\Omega_2(\Rsf(\Gscr))$ for all $p-2\geq k\geq 1$ \\
   since $e_{i_q i_{q+1} i_{q+2}}\in\Omega_2(\Rsf(\Gscr))$ for all $q\notin \{k-2, k-1\}$ follows from $e_{i_0...i_p}\in\Omega_p(\Rsf(\Gscr))$ by definition. 
   
   We establish each of (A) and (B) by explicitly constructing a triangulating pair for $e_{i_{k-2}i_{k-1}i_{k+1}}$ and $e_{i_{k-1} i_{k+1} i_{k+2}}$. We are required to consider three cases for each of the two constructive proofs. 
   We suggest that the reader refer to Figure \ref{fig:mainproof} to follow our proof with ease. 
   We will prove Case 1 of part (A) in detail and move rest to Appendix \ref{app:restcases}.\\
		(A)  $e_{i_{k-2}i_{k-1}i_{k+1}}\in\Omega_{2}(\Rsf(\Gscr))$ for all $p-1 \geq k\geq 2$.\\
		Note that since $e_{i_0...i_p}\in\Omega_p(\Rsf(\Gscr))$, $e_{i_{k-2}i_{k-1} i_k}$ and $e_{i_{k-1}i_{k} i_{k+1}}$ are both triangles. So let $(\Rscr_{\alpha_1}, \Rscr_{\beta_1})$ and $(\Rscr_{\alpha_2}, \Rscr_{\beta_2})$ be the respective triangulating pairs. \\
		Case 1: $\Rscr_{\alpha_1}^{i_{k-1}\to i_k}\cap \Rscr_{\alpha_2}^{i_{k-1}\to i_k}=\phi$. Let 
        \begin{align}
        \Rscr_\alpha &:=\Rscr_1^{o\to i_{k-2}}\cup\Rscr_{\alpha_1}^{i_{k-2}\to i_k}\cup \Rscr_{\beta_2}^{i_k\to i_{k+1}}\cup\Rscr_2^{i_{k+1}\to d},\\
		\Rscr_\beta &:=\Rscr_1^{o\to i_{k-2}}\cup\Rscr_{\beta_1}^{i_{k-2}\to i_{k-1}}\cup \Rscr_{\alpha_2}^{i_{k-1}\to i_{k+1}}\cup\Rscr_2^{i_{k+1}\to d}.\nonumber\label{eq:caseA1}
        \end{align}
		\textbf{Claim:} $(\Rscr_\alpha, \Rscr_\beta)$ is a triangulating pair for $e_{i_{k-2}i_{k-1}i_{k+1}}$. \\
  \textbf{Justification:} Clearly, $i_{k-2}, i_{k+1}\in \Rscr_\alpha$ so that $e_{i_{k-2 i_{k+1}}}\in \Rscr_\alpha \implies$ $\alpha\in\Cscr(e_{i_{k-2 i_{k+1}}})$. Similarly, we observe that $\beta\in\Cscr(e_{i_{k-2}i_{k-1}})\cap \Cscr(e_{i_{k-1}i_{k+1}})$. 
		Now let $j\in \Rscr_\alpha^{i_{k-2}\to i_{k+1}}=\Rscr_{\alpha_1}^{i_{k-2}\to i_k}\cup \Rscr_{\beta_2}^{i_k\to i_{k+1}}$. 
		If $j\in \Rscr_{\alpha_1}^{i_{k-2}\to i_k}$ and $j\neq i_{k-2}$, then $j\notin \Rscr_{\beta_1}^{i_{k-2}\to i_{k-1}}$ since $(\Rscr_{\alpha_1}, \Rscr_{\beta_1})$ is a triangulating pair for $e_{i_{k-2}i_{k-1}i_k}$ and $j\notin \Rscr_{\alpha_2}^{i_{k-1}i_{k+1}}$ by assumption of Case 1. On the other hand, if $j\in \Rscr_{\beta_2}^{i_{k}\to i_{k+1}}$ and $j\neq i_{k+1}$, then $j\notin \Rscr_{\beta_1}^{i_{k-2}\to i_{k-1}}$ since $j\geq i_k > i_{k-1}$ and $j\notin \Rscr_{\alpha_2}^{i_{k-1}i_{k+1}}$  since  $(\Rscr_{\alpha_2}, \Rscr_{\beta_2})$ is a triangulating pair for $e_{i_{k-1}i_{k}i_{k+1}}$. This establishes our claim.\\
		Case 2: $\Rscr_{\alpha_1}^{i_{k-1}\to i_k}\cap \Rscr_{\alpha_2}^{i_{k-1}\to i_k}=\bigcup_{n=1}^{n_0}\Rscr_{\alpha_1}^{p_n\to q_n}=\bigcup_{n=1}^{n_0}\Rscr_{\alpha_2}^{p_n\to q_n}\neq \phi$. (Recall Equation \eqref{eq:intersectiongraph}). Note that $i_k\notin \Rscr_{\alpha_2}$ and hence $q_{n_0}<i_k$.
			\end{proof}
	\section{Robust Paths in Series Parallel Graphs}\label{sec:spalgebra}
 We build a formal definition of a series-parallel graph in Subsection \ref{sec:spformal} and investigate chain complexes induced by them in Subsection \ref{sec:spcomplexes}. We find that $\dim\Omega_p(\Rsf(\Gscr))=0$ for all $p>2$ if and only if $\Gscr$ is a series-parallel graph which presents a notable correspondence between the developed homological algebra and emergent combinatorial complexity as $\Gscr$ deviates from a series-parallel topology. 
\subsection{Series-parallel two-terminal DAGs}\label{sec:spformal}
 We define parallel and series combinations below (Recall Figure \ref{fig:collection1}), following which an inductive definition for a series-parallel graph follows.  
	\begin{definition}\label{def:seriesparallel}
 Let $\Gscr_1=(\Vscr_1, \Escr_1, o_1, d_1)=\bigcup_i \Rscr^1_i$ and $\Gscr_2=(\Vscr_2, \Escr_2, o_2, d_2)=\bigcup_j \Rscr^2_j$ be two-terminal graphs.\\
		(i) If $\Gscr_1$ and $\Gscr_2$ satisfy $d_1=o_2, \Vscr_1\cap\Vscr_2=\{d_1\}$,  then a series combination of $\Gscr_1$ and $\Gscr_2$ is the two-terminal graph $\Gscr_1\to\Gscr_2:=\Gscr_1\cup\Gscr_2=(\Vscr, \Escr, o_1, d_2)=\bigcup_{i, j}\left( \Rscr^1_i\to\Rscr^2_j\right)$ where
		$\Vscr=\Vscr_1\cup \Vscr_2;\ \Escr=\Escr_1 \cup\Escr_2$.\\
		(ii) If $\Gscr_1$ and $\Gscr_2$ satisfy $o_1=o_2=:o$, $d_1=d_2=:d$ and  $\Vscr_1\cap\Vscr_2=\{o, d\}$,
		then, a parallel combination $\Gscr_1$ and $\Gscr_2$ is the two-terminal graph $\Gscr_1 || \Gscr_2:=\Gscr_1\cup\Gscr_2 = (\Vscr, \Escr, o, d)=\bigcup_{j\in \{1, 2\}i\in ([N_j])} \Rscr_i^j$ where
		$\Vscr=\Vscr_1\cup \Vscr_2;\ \Escr=\Escr_1\cup\Escr_2$.\\
  (iii) A two-terminal graph $\Gscr$ is a series-parallel graph if and only if 1) $\Gscr\cong K_{12}$ or
			2) $\Gscr\cong \Gscr_1\to\Gscr_2$ for series-parallel graphs $\Gscr_1$ and $\Gscr_2$ or 
			3) $\Gscr\cong \Gscr_1||\Gscr_2$ for series-parallel graphs $\Gscr_1$ and $\Gscr_2$.
	\end{definition}
	A series-parallel graph $\Gscr$ can hence be represented as a series and parallel combination of edges. For instance, the  graph $\Gscr_1$ in Figure \ref{fig:collection1}(a) expressible as follows in an `edge-combinatorial' representation:
	\begin{equation}
		\Gscr_1\to\Gscr_2= ((K_{i_0i_1}\to K_{i_1 i_3})||(K_{i_0i_2}\to K_{i_2 i_3})). \label{eq:algrep}
	\end{equation}
	
	\subsection{Path complexes of series-parallel graphs}\label{sec:spcomplexes}
	Given $\Gscr_1$ and $\Gscr_2$ along with their respectively induced chain complexes $\{\Omega_k(\Rsf(\Gscr_1))\}_{k\in\Nbb_0}$ and $\{\Omega_m(\Rsf(\Gscr_1))\}_{m\in\Nbb_0}$, we state what can be inferred about the complex induced by their combinations in the two propositions that follow.
	\begin{proposition}\label{prop:parallelhomology}
		Let $\Gscr_1=(\Vscr_1, \Escr_1, o, d)=\bigcup_{i\in [N_1]} \Rscr_i^1$ and $\Gscr_2=(\Vscr_2, \Escr_2, o, d)=\bigcup_{j\in[N_2]} \Rscr^2_j$, and, $\Gscr=(\Vscr, \Escr, o, d)=\Gscr_1 || \Gscr_2=\bigcup_{j\in \{1, 2\}i\in ([N_j])} \Rscr_i^j\equiv\bigcup_{\alpha\in [N]} \Rscr_\alpha$ be their parallel combination. Further, let $\Rsf(\Gscr_i)=(\Vscr_i, \Rsf(\Escr_i), \Cscr_i)$  for each $i=1, 2$ and $\Rsf(\Gscr)=(\Vscr, \Rsf(\Escr), \Cscr)$. Following relations then hold.\\
		(i) $\dim\Omega_0(\Rsf(\Gscr))=\dim\Omega_0(\Rsf(\Gscr_1))+\dim\Omega_0(\Rsf(\Gscr_2))-2.$\smallskip\\
		(ii)  $\dim \Omega_1(\Rsf(\Gscr))=\dim \Omega_1(\Rsf(\Gscr_1))+\dim \Omega_1(\Rsf(\Gscr_2))-1.$\smallskip\\
		(iii)  $\Omega_2(\Rsf(\Gscr))\supseteq \Omega_2(\Rsf(\Gscr_1))\cup \Omega_2(\Rsf(\Gscr_2))$.\smallskip\\
			(iv)  $\dim \Omega_p(\Rsf(\Gscr))=\dim \Omega_p(\Rsf(\Gscr_1))+\dim \Omega_p(\Rsf(\Gscr_2)), p>2$.
	\end{proposition}
	\begin{proof}
		(i) Follows since $\Omega_0(\Rsf(\Gscr))$ is a linear space spanned by all vertices of $\Rsf(\Gscr)$: $\dim \Omega_0(\Rsf(\Gscr))=|\Vscr_1\cup\Vscr_2|=|\Vscr_1|+|\Vscr_2|-2=|\dim\Omega_0(\Rsf(\Gscr_1))|+|\dim\Omega_0(\Rsf(\Gscr_2))|-2$.\\
		(ii) Follows since $\Omega_{1}(\Rsf(\Gscr))$ is a linear space spanned by all edges of $\Rsf(\Gscr)$: $\dim \Omega_1(\Rsf(\Gscr))=|\Rsf(\Escr_1)\cup\Rsf(\Escr_2)|=|\Rsf(\Escr_1)|+|\Rsf(\Escr_2)|-|\Rsf(\Escr_1)\cap\Rsf(\Escr_2)|=\dim \Omega_1(\Rsf(\Gscr_1))+\dim \Omega_1(\Rsf(\Gscr_2))-|\{e_{od}\}|=\dim \Omega_1(\Rsf(\Gscr_1))+\dim \Omega_1(\Rsf(\Gscr_2))-1$.\\
		(iii) Note that if $\Rscr\in \{\Rscr_\alpha\}_{\alpha\in [N_j]}$ for $j=1, 2$, then $\Rscr\in \{\Rscr_\alpha\}_{\alpha \in [N]}$. Thus, if a pair  $(\Rscr_{\alpha}, \Rscr_{\beta})$ triangulates $e_{i_0 i_1 i_2}$ in $\Rsf(\Gscr_j)$ for $j\in\{1, 2\}$, then it also triangulates the 2-path in $\Rsf(\Gscr)$. It follows that $\Delta_2(\Rsf(\Gscr_j))\subset \Delta_2(\Rsf(\Gscr))\implies\Omega_2(\Rsf(\Gscr_j))\subseteq \Omega_2(\Rsf(\Gscr))$ for each $j$, and thus, (iii) holds. \\
		(iv) Let $j\in \{1, 2\}$ and $e_{i_0...i_p}\in\Omega_p(\Rsf(\Gscr_j))$. Then $e_{i_{k-1} i_{k} i_{k+1}}\in\Omega_2(\Rsf(\Gscr_j))\implies e_{i_{k-1} i_{k} i_{k+1}}\in\Omega_2(\Rsf(\Gscr))$ for each $k\in [p-1]$ which in turn implies $e_{i_0...i_p}\in\Omega_p(\Rsf(\Gscr))$. Thus, $\Omega_p(\Rsf(\Gscr_1))\oplus\Omega_p(\Rsf(\Gscr_2))\subset \Omega_p(\Rsf(\Gscr))$. 
		On the other hand, if $e_{i_{k-1} i_k i_{k+1}}$ does not have a triangulating pair in $\Rsf(\Gscr_j)$, then it cannot not have one in $\Rsf(\Gscr)$ either unless $(i_{k-1}, i_{k+1})\neq (o, d)$ since then atleast one of $i_{k-1}$ and $i_{k+1}$ does not belong $\Gscr_{-j}$ (where $-j\in\{1, 2\}, j\neq -j$). Hence, $\Omega_p(\Rsf(\Gscr_1))\oplus\Omega_p(\Rsf(\Gscr_2))= \Omega_p(\Rsf(\Gscr))$ follows.
		Further, since $p>2$ and $\Vscr_1\cap \Vscr_2=\{o, d\}$, $\Omega_p(\Rsf(\Gscr_1))\perp\Omega_p(\Rsf(\Gscr_2))$ and the proposed follows.
	\end{proof}
	
	\begin{proposition}\label{prop:serieshomology}
		Let $\Gscr_1=(\Vscr_1, \Escr_1, o, h)$ and $\Gscr_2=(\Vscr_2, \Escr_2, h, d)$ be two two-terminal graphs and $\Gscr=\Gscr_1 \to \Gscr_2=(\Vscr, \Escr, o, d)$ be their series combination. Then, the following hold.\\
		(i) $\dim \Omega_0(\Rsf(\Gscr))=\dim \Omega_0(\Rsf(\Gscr_1))+\dim \Omega_0(\Rsf(\Gscr_2))-1.$\\
		(ii) $\dim \Omega_p(\Rsf(\Gscr))=\dim \Omega_p(\Rsf(\Gscr_1))+\dim \Omega_p(\Rsf(\Gscr_2))\ \forall p>1 .$
	\end{proposition}
	\begin{proof}
		(i) Follows since $\Omega_0(\Rsf(\Gscr))$ is a linear space spanned by all vertices of $\Rsf(\Gscr)$: $\dim \Omega_0(\Rsf(\Gscr))=|\Vscr_1\cup\Vscr_2|=|\Vscr_1|+|\Vscr_2|-1=|\dim\Omega_0(\Rsf(\Gscr_1))|+|\dim\Omega_0(\Rsf(\Gscr_2))|-1$.\\
		(ii) For $p=1$ follows since $\Omega_1(\Rsf(\Gscr))$ is a linear space spanned by all edges of $\Rsf(\Gscr)$: $\dim \Omega_1(\Rsf(\Gscr))=|\Escr_1\cup\Escr_2|=|\Vscr_1|+|\Vscr_2|=|\dim\Omega_1(\Rsf(\Gscr_1))|+|\dim\Omega_1(\Rsf(\Gscr_2))|$. Now notice that no pair of routes can triangulate $e_{j h k}$ for all $j, k\in \Vscr$ with $j<h<k$ since $h$ belongs to every route of $\Gscr$ by definition. Thus $e_{i_0...i_p}\in\Omega_p(\Rsf(\Gscr))$ requires $h\leq i_0$ or $h\geq i_p$ which is equivalent to $e_{i_0...i_p}\in \Omega_p(\Rsf(\Gscr_2))$ or $e_{i_0...i_p}\in \Omega_p(\Rsf(\Gscr_1))$ respectively. Thus, $\Omega_p(\Rsf(\Gscr))=\Omega_p(\Rsf(\Gscr_1))\oplus\Omega_p(\Rsf(\Gscr_2))$ and the proposed follows.
	\end{proof}
\noindent We are now in a position to establish our main result. 
\begin{theorem}\label{thm:spbifurcation}
	Let $\Gscr=(\Vscr, \Escr, o, d)=\cup_{\alpha\in A} \Rscr_\alpha$ be a two terminal DAG. Then,\\
	(i) If $\Gscr$ is a series-parallel graph, then $\dim \Omega_p(\Gscr)=0$ for all $p\geq 3$.\\
	(ii) If $\Gscr$ is not a series-parallel graph, then $\dim\Omega_3(\Rsf(\Gscr))>0$.
\end{theorem}
\begin{proof}
 (i) If $\Gscr$ is a series-parallel graph, then using Propositions \ref{prop:parallelhomology}.(iv), \ref{prop:serieshomology}.(ii), and the edge combinational representation of $\Gscr=(\Vscr, \Escr, o, d)$, we deduce
 \begin{multline}
\dim\Omega_p(\Rsf(\Gscr))=\sum_{e_{ij}\in \Escr} \dim\Omega_p(\Rsf(K_{ij}))=0\ \forall\ p\geq 3.
 \end{multline}
 (ii) If $\Gscr$ is not series-parallel, then sequentially decomposing $\Gscr$ serially and/or parallely one eventually arrives at a two-terminal subgraph $\Gscr'=(\Vscr', \Escr', o', d')\neq K_{o'd'}$ which is not decomposable further. Since $\Gscr'\subseteq \Gscr$, we have an $A'\subseteq A$ such that $\Gscr'=\cup_{\alpha\in A'} \Rscr_\alpha^{o'\to d'}$. Note that $|A'|>1$ since $|A'|=1$ implies $\Gscr'$ is a single $o'\to d'$ route which is serially decomposable by definition.\\
 \emph{Case 1:} We show that if $o', d'$ are vertices of some robust 2-path in $\Rsf(\Gscr')$, then $\Gscr'$ is parallely decomposable without a robust 3-path in it. Appendix \ref{app:case1p2thm} sketches the complete proof.\\
\emph{Case 2:} We will show that if there is  no robust 2-path in $\Gscr'$ with $o', d'$ as its vertices, then there must be a robust 3-path in $\Gscr'$ to avoid (serial/parallel) decomposability  of $\Gscr'$. See Appendix \ref{app:case2p2thm} for complete proof of this case.
The proof thereby rests as we show that if $\Gscr'$ is indecomposable as supposed, then it must contain a robust 3-path and thereby $\Gscr$ contains one too. 
\end{proof}
\subsection{Robust 3-paths and the Braess paradox}
We say that a graph $\Hscr$ is embedded in a graph $\Gscr$ if upon deletion of suitable edges and vertices in $\Gscr$, and subsequent merging of edges $e_{ij}$ and $e_{jk}$ in the graph obtained upon the deletion into a single edge $e_{ik}$, the result is a graph $\Gscr'$ that is isomorphic to $\Hscr$. For example, for $\Gscr$ in Figure \ref{fig:dim3unnecessary}, deleting $e_{35}, e_{5d}$ and merging the pairs $e_{o1}, e_{12}$ to $e_{o2}$ and $e_{o3}, e_{34}$ to $e_{o4}$ yields a graph isomorphic to the Braess embedding in Figure \ref{fig:collection1}(b) with $(i_0, i_1, i_2, i_3)\cong (o,2, 4, d)$. 

Any robust 3-path $e_{i_{k-2}i_{k-1} i_{k} i_{k+1}}$ in $\Gscr$ is induced as one of the three cases in Figure \ref{fig:mainproof}. In Case 1, the vertex tuple $(i_{k-2}, i_{k-1}, i_{k}, i_{k+1})$ induces a Braess embedding. Similarly, in Case 2(a), the vertex tuple $(i_{k-2}, v_{\ell_0}, q_{n_0}, i_{k+1})$ induces an embedding and in Case 2(b), the tuple $(i_{k-2}, i_{k-1}, q_{n_0}, i_{k+1})$ does so. A robust $p$-path $e_{i_0...i_p}$ with $p>3$ contains ${p \choose 4}$ robust 3-paths within itself and hence identifies a large collection of Braess-susceptible sites. Conversely, if $(i_0, i_1, i_2, i_3)$ induce a Braess embedding in $\Gscr$, then reintroduction of all deleted edges and vertices in the embedding reconstructs $\Gscr$ that contains the robust 3-path $e_{i_0i_1i_2i_3}$ with the structure shown by Case 1 in Figure \ref{fig:mainproof}.

	\section{Conclusion}\label{sec:conclusion}
In this paper, we introduced the notion of a $k$-robust path in a DAG $\Gscr$, which localizes the increasing deviation of a graph $\Gscr$ from a series-parallel topology for increasing values of $k$. We showed that the association of the $\mathbb K$-linear spaces of robust $k$-paths with $k$-chains in a chain complex sets up a consistent graph homology. We established that the topological simplicity of series-parallel graphs translates into a triviality of $k$-chains in the induced complex for $k\geq 3$, and any non-triviality therein deviates the graph from the simple topology. We further discussed the resulting correspondence between the space of $3$-chains and Braess-susceptible sites within a network. With this discussion serving as an illustrative example, we believe that the graph homology developed with robust paths as its basis will be a useful tool for the systematic characterization of complex behavior in flow networks.
	\bibliographystyle{unsrt}
	\bibliography{ref}

 \appendix
 
 \subsection{Proof of Proposition \ref{prop:homologysetup} continued.}\label{app:restcases}
 We list rest of the cases below and construct respective triangulating pairs alongside.
\noindent (A)  $e_{i_{k-2}i_{k-1}i_{k+1}}\in\Omega_{2}(\Rsf(\Gscr))$ for all $p-1 \geq k\geq 2$.\\
		-- Case 2(a): $\Rscr_{\alpha_2}^{i_{k-1}\to i_{k+1}}\cap \Rscr_{\beta_1}^{i_{k-1}\to i_{k}}=\bigcup_{\ell=1}^{\ell_0}  \Rscr_{\alpha_2}^{u_\ell \to v_\ell}=\bigcup_{\ell=1}^{\ell_0}  \Rscr_{\beta_1}^{u_\ell \to v_\ell}$ and $v_{\ell_0}<q_{n_0}$.Let 
		$$\Rscr_\alpha:=\Rscr_1^{o\to i_{k-2}}\cup\Rscr_{\alpha_1}^{i_{k-2}\to q_{n_0}}\cup \Rscr_{\alpha_2}^{q_{n_0}\to i_{k+1}}\cup\Rscr_2^{i_{k+1}\to d},$$
		$$\Rscr_\beta:=\Rscr_1^{o\to i_{k-2}}\cup\Rscr_{\beta_1}^{i_{k-2}\to i_{k}}\cup \Rscr_{\beta_2}^{i_{k}\to i_{k+1}}\cup\Rscr_2^{i_{k+1}\to d}.$$
			-- Case 2(b): $\Rscr_{\alpha_2}^{i_{k-1}\to i_{k+1}}\cap \Rscr_{\beta_1}^{i_{k-1}\to i_{k}}=\bigcup_{\ell=1}^{\ell_0}  \Rscr_{\alpha_2}^{u_\ell \to v_\ell}=\bigcup_{\ell=1}^{\ell_0}  \Rscr_{\beta_1}^{u_\ell \to v_\ell}$ and $v_{\ell_0}>q_{n_0}$.
			Let
			$$\Rscr_\alpha:=\Rscr_1^{o\to i_{k-2}}\cup\Rscr_{\alpha_1}^{i_{k-2}\to i_k}\cup \Rscr_{\beta_2}^{i_k\to i_{k+1}}\cup\Rscr_2^{i_{k+1}\to d},$$
		$$\Rscr_\beta:=\Rscr_1^{o\to i_{k-2}}\cup\Rscr_{\beta_1}^{i_{k-2}\to v_{\ell_0}}\cup \Rscr_{\alpha_2}^{v_{\ell_0}\to i_{k+1}}\cup\Rscr_2^{i_{k+1}\to d}.$$
		(B) $e_{i_{k-1} i_{k+1} i_{k+2}}\in\Omega_2(\Rsf(\Gscr))$ for all $p-2\geq k\geq 1$.\\
		 As before, note that $e_{i_{k-1}i_{k}i_{k+1}}$ and $e_{i_{k}i_{k+1}i_{k+2}}$ are both triangles owing to $e_{i_0...i_p}\in\Omega_p(\Rsf(\Gscr))$ and let $(\Rscr_{\alpha_1}, \Rscr_{\beta_1})$ and $(\Rscr_{\alpha_2}, \Rscr_{\beta_2})$ be the respective triangulating pairs. Our proof scheme remains exactly as before. We will construct a triangulating pair $(\Rscr_\alpha, \Rscr_\beta)$ for the 2-path $e_{i_{k-1}i_{k+1}i_{k+2}}$.\\ 
		Case 1: $\Rscr_{\alpha_1}^{i_{k-1}\to i_{k+1}}\cap \Rscr_{\alpha_2}^{i_{k}\to i_{k+2}}=\phi$.
		$$\Rscr_\alpha= \Rscr_1^{o\to i_{k-1}}\cup\Rscr_{\beta_1}^{i_{k-1}\to i_{k}}\cup \Rscr_{\alpha_2}^{i_k\to i_{k+2}}\cup\Rscr_2^{i_{k+2\to d}},$$
			$$\Rscr_\beta= \Rscr_1^{o\to i_{k-1}}\cup\Rscr_{\alpha_1}^{i_{k-1}\to i_{k+1}}\cup \Rscr_{\beta_2}^{i_{k+1}\to i_{k+2}}\cup\Rscr_2^{i_{k+2\to d}}.$$
			Case 2: $\Rscr_{\alpha_1}^{i_{k-1}\to i_{k+1}}\cap \Rscr_{\alpha_2}^{i_{k}\to i_{k+2}}=\bigcup_{n=1}^{n_0} \Rscr_{\alpha_1}^{p_n\to q_n}=\bigcup_{n=1}^{n_0} \Rscr_{\alpha_2}^{p_n\to q_n}$.\\
			(a)  $\Rscr_{\alpha_2}^{i_{k}\to i_{k+2}}\cap \Rscr_{\beta_1}^{i_{k}\to i_{k+1}}=\bigcup_{\ell=1}^{\ell_0} \Rscr_{\alpha_2}^{u_\ell \to v_\ell}$, $v_{\ell_0}< q_{\ell_0}$.
			$$\Rscr_\alpha=\Rscr_1^{o\to i_{k-1}}\cup \Rscr_{\alpha_1}^{i_{k-2}\to q_{n_0}}\cup \Rscr_{\alpha_2}^{q_{n_0}\to i_{k+2}}\cup \Rscr_{2}^{i_{k+2}\to d},$$
				$$\Rscr_\beta=\Rscr_1^{o\to i_{k-1}}\cup \Rscr_{\beta_1}^{i_{k-1}\to _{k+1}}\cup \Rscr_{\beta_2}^{i_{k+1}\to i_{k+2}}\cup \Rscr_{2}^{i_{k+2}\to d}.$$
				(b) $\Rscr_{\alpha_2}^{i_{k}\to i_{k+2}}\cap \Rscr_{\beta_1}^{i_{k}\to i_{k+1}}=\bigcup_{\ell=1}^{\ell_0} \Rscr_{\alpha_2}^{u_\ell \to v_\ell}$, $v_{\ell_0}> q_{\ell_0}$.
				$$\Rscr_\alpha=\Rscr_1^{o\to i_{k-1}}\cup \Rscr_{\beta_1}^{i_{k-2}\to v_{\ell_0}}\cup \Rscr_{\alpha_2}^{v_{\ell_0}\to i_{k+2}}\cup \Rscr_{2}^{i_{k+2}\to d},$$
				$$\Rscr_\beta=\Rscr_1^{o\to i_{k-1}}\cup \Rscr_{\alpha_1}^{i_{k-1}\to _{k+1}}\cup \Rscr_{\beta_2}^{i_{k+1}\to i_{k+2}}\cup \Rscr_{2}^{i_{k+2}\to d}.$$
				It follows that $e_{i_0..._p}\in\Omega_{p}(\Rsf(\Gscr))\implies e_{i_0...\widehat i_k...i_p}\in\Omega_{p-1}(\Rsf(\Gscr))$. This completes the proof of the proposition.
    \subsection{Proof of Theorem \ref{thm:spbifurcation} continued}
    \subsubsection{Case 1, Part (ii)}\label{app:case1p2thm}
     Suppose that $\exists\ k\in \Vscr': e_{o'kd'}\in\Omega_2(\Rsf(\Gscr'))$ and let $(\Rscr_\alpha^{o'\to d'}, \Rscr_\beta^{o'\to d'})$ be the corresponding triangulating pair whereby $\Rscr_\alpha^{o'\to d'}\cap \Rscr_\beta^{o'\to d'}=\{o', d'\}$ (Recall Definition \ref{def:invpath}).
 Define a subset $D'\subset A'$:
$D'=\{\delta: \Rscr_\delta^{o'\to d'}\cap \Rscr_\alpha^{o'\to d'}=\{o', d'\}$
 and $D^{\prime c}=A'\setminus D'$. Note that $\beta\in D'$ and $\alpha\in D'^c$ so that $D', D'^c$ are both non-empty. Now consider the following two-terminal graphs induced by the partition $\{D', D'^c\}$:
 $$\Gscr'_{D'}=\bigcup_{D'}\Rscr_\delta^{o'\to d'}=(\Vscr'_{D'},\Escr'_{D'}, o', d');$$
 $$\Gscr'_{D'^c}=\bigcup_{D'^c}\Rscr_\delta^{o'\to d'}=(\Vscr'_{D'^c},\Escr'_{D'^c}, o', d').$$
 If $\Vscr_{D'}'\cap \Vscr'_{D'^c}=\{o', d'\}$
then $\Gscr'=\Gscr_{D'}||\Gscr_{D'^c}$ which contradicts the supposition that $\Gscr'$ is not decomposable parallely. Otherwise if $ j\in \Vscr_{D'}'\cap \Vscr'_{D'^c}, j\notin \{o', d'\}$ then $j\in \Rscr_{\delta_0}^{o'\to d'}$ for some $\delta_0\in D'^c$ and $j\notin \Rscr_\alpha^{o'\to d'}$. Further, choose any $\gamma\in D'$ with $j\in\Rscr_\gamma^{o'\to d'}$ ($\gamma$ exists since $j\in\Vscr_{D'}$). Then, atleast one of the following two vertices exist outside $\{o', d', j\}$:
$$\ell_1=\min{\Rscr_{\delta_0}^{j\to d'}\cap \Rscr_\alpha^{o'\to d'}},\ \ell_2=\max{\Rscr_{\delta_0}^{o'\to j}\cap \Rscr_\alpha^{o'\to d'}}.$$
If $\ell_1$ exists outside $\{o', d', j\}$ then $e_{o'j\ell_1 d'}\in\Omega_3(\Rsf(\Gscr'))$ as $e_{o' j \ell_1}$ is triangulated by the pair $(\Rscr^{o'\to d'}_a, \Rscr^{o'\to d'}_b)$ given by 
$$\Rscr^{o'\to d'}_a=\Rscr_\alpha^{o'\to d'}; \ \Rscr^{o'\to d'}_b=\Rscr_\gamma^{o'\to j}\cup \Rscr_{\delta_0}^{j\to d'}$$
and $e_{j\ell_1 d'}$ is triangulated by $(\Rscr^{o'\to d'}_c, \Rscr^{o'\to d'}_d)$ given by 
$$\Rscr_c^{o'\to d'}=\Rscr_\gamma^{o'\to d'}, \Rscr_d^{o'\to d'}=\Rscr_{\delta_0}^{j\to \ell_1}\cup\Rscr_{\alpha}^{\ell_1\to d'}$$
as can be verified considering $\Rscr_\alpha^{o'\to d'}\cap\Rscr_{\beta}^{o'\to d'}=\{o', d'\}$. Similarly if $\ell_2$ exists outside $\{o', d', j\}$, $e_{o'\ell_2 j d'}\in\Omega_3(\Rsf(\Gscr'))$ as $e_{o'\ell_2 j}$ is triangulated by $(\Rscr_a^{o'\to d'}, \Rscr_b^{o'\to d'})$ given by 
$$\Rscr_a^{o'\to d'}=\Rscr_\gamma^{o'\to d'},$$ $$\Rscr_b^{o'\to d'}=\Rscr_\alpha^{o'\to \ell_2}\cup\Rscr_{\delta_0}^{\ell_2\to d'}$$
and $e_{\ell_2 j d'}$ is triangulated by $(\Rscr_c^{o'\to d'}, \Rscr_d^{o'\to d'})$ given by 
$$\Rscr_c^{o'\to d'}=\Rscr_\alpha^{o'\to d'},$$ $$\Rscr_d^{o'\to d'}=\Rscr_{\delta_0}^{o'\to j}\cup\Rscr_{\gamma}^{j\to d'}.$$
Thus, $\dim\Omega_3(\Rsf(\Gscr'))>0$ and we arrive at a contradiction since $0=\dim\Omega_3(\Rsf(\Gscr))>\dim\Omega_3(\Rsf(\Gscr'))$ following $\Gscr'\subset\Gscr\implies \Omega_3(\Rsf(\Gscr'))\subset \Omega_3(\Rsf(\Gscr))$.
    \subsubsection{Case 2, Part (ii)}\label{app:case2p2thm}
    Suppose that $\nexists\ k\in \Vscr: e_{o'k d'}\in\Omega_2(\Rsf(\Gscr'))$.
If $\Gscr'$ has only two routes, \ie, $A'=\{a, b\}$, then, $\Rscr^{o'\to d'}_a\cap\Rscr^{o'\to d'}_b = \{o', d'\}$ will ensure that $(\Rscr^{o'\to d'}_a, \Rscr^{o'\to d'}_b)$ triangulate some $e_{o' k d'}$ for some $k\in \Rscr_a^{o'\to d'}\cup\Rscr_b^{o'\to d'}$, a contradiction to the presupposition of Case 2. So let $k\in \Rscr^{o'\to d'}_a\cap \Rscr_b^{o'\to d'}$. This allows $\Gscr'$ the serial decomposition
$\Gscr'=\Rscr_a^{o'\to k}\cup\Rscr_b^{o'\to k} \to \Rscr_a^{k\to d'}\cup\Rscr_b^{k\to d'}$.\\
Now let $|A'|>2$. If there is a vertex $k$ common across all routes, \ie,  $\exists k\in\Vscr'$ such that $k\in\Rscr_a^{o'\to d'}$ for all $a\in A'$, $\Gscr'$ admits the serial decomposition $\Gscr'=\bigcup_{A'}\Rscr_a^{o'\to k} \to \bigcup_{A'}\Rscr_a^{k\to d'}$ and we arrive at a contradiction. Now suppose otherwise so that for each $k\in\Vscr'\setminus\{o', d'\}$, there is a route in $\Gscr'$ that excludes $k$. We will show that this supposition contradicts atleast one of $\dim\Omega_3(\Rsf(\Gscr'))=0$ or the supposition of Case 2 \ie $\nexists j: e_{o'jd'}\in \Omega_2(\Rsf(\Gscr'))$. 

We begin by establishing the existence of a pair $j_0, k_0\in \Vscr'$ such that $e_{o' j_0 k_0}\in\Omega_2(\Rsf(\Gscr'))$. Pick an arbitrary $\beta\in A'$ and take $j_0$ as the second vertex in the route $\Rscr_\beta^{o'\to d'}$ \ie $j_0=\min \Rscr_{\beta}^{o'\to d'}$ such that $j_0\neq o'$. Then there is a route $\Rscr_\alpha^{o'\to d'}$ that excludes $j_0$. Take $k_0=\min \Rscr_\alpha^{o'\to d'}\cap \Rscr_{\beta}^{o'\to d'}, k_0>o'$. Then $k_0>j_0$ since $j_0$ is the second smallest vertex in $\Rscr_\beta^{o'\to d'}$. This selection ensures $e_{o'j_0k_0}\in \Omega_2(\Rsf(\Gscr))$ with the triangulating pair $(\Rscr_\alpha^{o'\to d'}, \Rscr_\beta^{o'\to d'})$.

Now let $e_{o' j k}\in\Omega_2(\Rsf(\Gscr'))$ with arbitrary $j, k\in\Vscr'$ with corresponding triangulating pair $(\Rscr_\alpha, \Rscr_\beta)$. By supposition, there is a route $\Rscr_\gamma^{o'\to d'}$ that excludes $k$. \\
-- Let
$\ell_\alpha=\min\Rscr_\alpha^{o'\to d'}\cap \Rscr_\gamma^{o'\to d'}$ such that $\ell_\alpha=o'$ and note that if $\ell_\alpha>k_0$, then $e_{o'k_0\ell_\alpha}\in\Omega_2(\Rsf(\Gscr'))$ with triangulating pair $(\Rscr_\gamma^{o'\to d'}, \Rscr_\alpha^{o'\to d'})$. \\
-- Let
$\ell_\beta=\min\Rscr_\beta^{o'\to d'}\cap \Rscr_\gamma^{o'\to d'}$ such that $\ell_\beta=o'$ and note that if $\ell_\beta>k_0$, then $e_{o'k_0\ell_\beta}\in\Omega_2(\Rsf(\Gscr'))$ with triangulating pair $(\Rscr_\gamma^{o'\to d'}, \Rscr_\beta^{o'\to d'})$. \\
-- If both $\ell_\alpha<k_0$ and $\ell_\beta<k_0$, then the intersection graphs $\Rscr_\gamma ^{o'\to d'}\cap \Rscr_\alpha^{o'\to k_0} $ and $\Rscr_\gamma ^{o'\to d'}\cap \Rscr_\beta^{o'\to k_0} $ contain other vertices in addition to $o'$ so let $p_\alpha$ and $p_\beta$ be the maximal vertices in the above two intersection graphs and note that $p_\alpha<k_0$ and $p_\beta<k_0$ hold by definition.\\
-- Now if $p_\alpha < p_\beta$, then $e_{o'p_\alpha p_\beta k_0}\in \Omega_3(\Rsf(\Gscr'))$ since

-- $e_{o' p_\alpha p_\beta}$ is triangulated by $(\Rscr_\beta^{o'\to d'}, \Rscr_\alpha^{o'\to p_\alpha}\cup \Rscr_\gamma^{p_\alpha\to d'})$.

-- $e_{p_\alpha p_\beta k_0}$ is triangulated by $(\Rscr_\alpha^{o'\to d'}, \Rscr_\beta^{o'\to p_\beta}\cup \Rscr_\gamma^{p_\beta\to d'})$.\\
-- Otherwise if $p_\alpha>p_\beta$, then  $e_{o'p_\beta p_\alpha k_0}\in \Omega_3(\Rsf(\Gscr'))$ since

-- $e_{o' p_\beta p_\alpha}$ is triangulated by $(\Rscr_\alpha^{o'\to d'}, \Rscr_\beta^{o'\to p_\beta}\cup \Rscr_\gamma^{p_\beta\to d'})$.

-- $e_{p_\beta p_\alpha k_0}$ is triangulated by $(\Rscr_\beta^{o'\to d'}, \Rscr_\alpha^{o'\to p_\alpha}\cup \Rscr_\gamma^{p_\alpha\to d'})$.\\

This long line of reasoning thus brings us to the following conclusion. If there is a route for every vertex in $\Vscr'\setminus\{o', d'\}$ that excludes it, 
then there exists atleast one 2-path $e_{o' jk}\in\Omega_2(\Rsf(\Gscr'))$. Further, existence of a 2-path $e_{o' jk}\in\Omega_2(\Rsf(\Gscr'))$ implies one of the following two implications:\\
--  Either $\dim\Omega_3(\Rsf(\Gscr'))\neq 0$ which is a contradiction to the presupposition on $\Gscr'$\\
-- Or there is a pair $ j', k'\in \Vscr'$ with $k'>k$ such that $e_{oj'k'}\in \Omega_2(\Rsf(\Gscr'))$. One can then set $j=j', k=k'$ and inductively run the same line of arguments again to obtain either $\dim\Omega_3(\Rsf(\Gscr'))\neq 0$ or $k'=d'$ which is a contradiction to the presupposition of Case 2.
The proof rests here.
	
\end{document}